\documentclass[11pt]{amsart}
\usepackage[utf8]{inputenc}
\usepackage{amsmath}
\usepackage{mathtools}
\usepackage[foot]{amsaddr}
\usepackage{graphicx}
\usepackage{amssymb}
\usepackage{mathrsfs}
\usepackage[nobysame, alphabetic]{amsrefs}
\usepackage{color}
\definecolor{darkblue}{rgb}{0,0,0.45}
\usepackage[colorlinks,linkcolor=darkblue,anchorcolor=green,citecolor=darkblue]{hyperref}
\usepackage[margin=1.25in]{geometry}
\usepackage{microtype}

\usepackage{etoolbox}
\patchcmd{\section}{\scshape}{\bfseries\large}{}{}
\makeatletter
\def\@seccntformat#1{\csname the#1\endcsname.\space}

\patchcmd{\abstract}{\scshape\abstractname}{\textbf{\abstractname}}{}{}
\makeatother

\makeatletter
\def\l@subsection{\@tocline{2}{0pt}{3em}{5em}{}}
\makeatother

\newtheorem{theorem}{Theorem}
\newtheorem{lemma}{Lemma}
\newtheorem{remark}{Remark}
\newtheorem{proposition}[theorem]{Proposition}

\newtheorem*{conjecture*}{Conjecture}
\newtheorem*{question*}{Question}
\newtheorem*{theorem*}{Theorem}

\numberwithin{definition}{section}

\renewcommand{\AA}{\mathbb{A}}

\newcommand{\CC}{\mathbb{C}}

\newcommand{\NN}{\mathbb{N}}

\newcommand{\PP}{\mathbb{P}}
\newcommand{\QQ}{\mathbb{Q}}
\newcommand{\RR}{\mathbb{R}}

\newcommand{\ZZ}{\mathbb{Z}}

\title{Deciding subspace reachability problems with application to Skolem's Problem}
\author{Samuel Everett}
\email{same@uchicago.edu}

\begin{document}

\begin{abstract}
The higher-dimensional version of Kannan and Lipton's Orbit Problem asks whether it is decidable if a target subspace can be reached from a starting point under repeated application of a linear transformation. Similarly, the continuous analog of the Orbit Problem asks if a flow induced by a linear system of differential equations ever reaches some specified subspace. The decidability of both problems remains open, and in fact the problems generalize the discrete and continuous versions of Skolem's Problem. The object of this paper is to communicate a geometric perspective of the discrete and continuous Orbit Problems, alternate to the traditional and highly technical algebraic and number-theoretic approaches to the problem. We derive a simple decision procedure capable of deciding a certain class of instances of the Orbit Problem, and, as an application, we obtain alternate proofs to a number of results using elementary geometric arguments.
\end{abstract}

\maketitle

\tableofcontents

\section{Introduction}

\subsection{Orbit Problem Overview}
In a pair of seminal papers \cites{kannan1980orbit,kannan1986polynomial} Kannan and Lipton introduced the \emph{Orbit Problem}, motivated by reachability questions of linear sequential machines raised by Harrison \cite{harrison1969lectures}. Given a linear transformation $A \in \mathbb{Q}^{d \times d}$, and elements $x,y \in \mathbb{Q}^d$, the Orbit Problem asks whether it can be decided if there exists an $n \in \NN$, such that $A^nx = y$. Kannan and Lipton proved decidability of this point-to-point version of the problem, but remarked that when the target is a subspace $U$ of $\QQ^d$, the problem becomes far more difficult.
One may additionally consider a \emph{continuous} version of the Orbit Problem, in which it is to be decided whether a flow induced by a linear system of differential equations reaches a target hyperplane, or whether real-valued exponential polynomials have any real roots \cite{bell2010continuous}.
More formally, presented with an ordinary differential equation $dx(t)/dt = Ax(t)$, with $A \in \mathbb{Q}^{d \times d}$, and initial condition $x(0) \in \mathbb{Q}^d$, it is to be decided whether there exists a $t \in \mathbb{R}^+$ such that $c^\top \exp(At)x(0) = 0$ for $c \in \mathbb{Q}^d$.
Despite the importance of these problems due to their connection with Skolem's Problem and the foundations of program verification and control theory, progress on the problems has been slow.

In the discrete-time setting, no progress was made on the higher-dimensional Orbit Problem until 2013, when in beautiful work Chonev, Ouaknine, and Worrell proved that the higher-dimensional Orbit Problem is decidable whenever the dimension of the target space is one, two, or three \cite{chonev2013orbit,chonev2016complexity}.
In following work, the same authors considered a version of the problem where the target space is not a subspace but rather a polytope \cite{chonev2014polyhedron}, proving decidability when the target is of dimension three or less, and hardness with respect to long-standing number theoretic open problems for higher dimensions.
In related work, decidability for dimensions up to three was also shown for generalized cases in which the source and target sets are both either polytopes \cite{almagorPolytopeCollision} or semi-algebraic sets \cite{almagorSemialgebraicOrbit}. One other prominent recent approach involves generating sets that are invariant under the action of the matrix $A$ that contain the starting point, and then proving such sets are disjoint from the target set \cite{de2018left,fijalkow2019complete,almagor2022minimal}. However, many difficulties in analyzing the Orbit Problem remain very-much alive, with most substantive progress being made conditioned on long-standing number theoretic conjectures. In particular, the decision procedures provided in these works rely on highly technical methods from number theory, model theory, and algebra.

The continuous-time setting offers no reprieve.
In fact there have been fewer results making progress on the continuous version of the Orbit problem; partially because the usual algebraic and number theoretic techniques do not carry over to the continuous time setting as easily.
Nevertheless, there have been a number of strong results toward decidability.
Early work of Hainry \cite{hainry2008reachability,hainry2009decidability} proved decidability of the continuous point-to-point Orbit Problem, a result which has since been strengthened \cite{dantam2021decidability,chen2015continuous}.
Bell et al. \cite{bell2010continuous} made substantial progress on the problem, particularly proving $\textsf{NP}$-hardness and decidability for dimensions two or less.
In another successful line of work, Chonev, Ouaknine, and Worrell proved decidability up to dimension eight conditioned on Schanuel's Conjecture, and proved decidability up to dimension eight unconditionally for a restricted version of the problem
\cite{chonev2016recurrent,chonev_et_al2016,10.1145/3603543}.
Despite these deep results toward the problem, there remain many obstacles to further progress.

\subsection{Statement of Results}

The aim of this paper is to present an approach to the discrete and continuous Orbit Problems centered around the observation that equivalent problems can be asked and answered in real-projective space rather than Euclidean space. Once this observation has been made, many technical theorems become relatively elementary to prove due to the substantial increase in geometric structure available.

This paper begins development in the discrete-time setting, with first application of our techniques to the discrete-time Orbit Problem and Skolem Problem. Subsequent application of our techniques in the context of flows is straight-forward.
Our first result is concerned with developing a simple geometric algorithm capable of deciding the problem when certain spectral and geometric constraints are satisfied. However, the following result is principally used as a foothold in proving all the other results in this paper.

\begin{theorem}\label{thmOne}
Let $(A, x, U)$ be any instance of the discrete-time Orbit Problem, where $A\in \QQ^{d \times d}$, $x \in \mathbb{Q}^d$ is non-trivial\footnote{We say a point $x$ is \emph{non-trivial} (w.r.t. a matrix $A$) if $x$ has a non-zero component with respect to the Jordan blocks of Jordan matrix $J(A)$. The non-triviality condition on $x$ assures the problem does not collapse to a lower-dimension; if the condition is removed an essentially identical statement holds.}, and $U$ is a target subspace in $\QQ^d$, described by rational parameters. 
Suppose $A$ has $r$ distinct eigenvalues $\lambda_1,\dots,\lambda_r$ of largest modulus. Let $n_i$ denote the multiplicity of the root $\lambda_i$, $i=1,\dots,r$, in the minimal polynomial of $A$. Suppose $n_1= \cdots = n_p$, $p \leq r$, are the largest such values.
Then there exists a subspace $W$ computable from $A$ and $x$ of dimension $p$, such that if $W \cap U = \{0\}$, it is decidable whether there exists $n \in \NN$ such that $A^nx \in U$.
\end{theorem}

Once the step of compactifying state space has been taken by translating the problem to projective space, proof of Theorem \ref{thmOne} is conceptually easy: it essentially reduces to a case analysis of the asymptotics of any orbit under iteration of a matrix in Jordan normal form.
In particular, Jordan blocks comprising the normal form of the matrix $A$ carry all the information needed to determine the subspace any given initial condition will develop toward, after compactifying by projecting onto real projective space.

Theorem \ref{thmOne} is to be primarily viewed as providing a \emph{tool} for determining decidability of problem classes. For instance, one immediate application of Theorem \ref{thmOne} is in proving a weakened version of the breakthrough result of Chonev, Ouaknine, and Worrell that the higher-dimensional Orbit Problem is decidable whenever the dimension of the target space is one \cite{chonev2013orbit,chonev2016complexity}. Namely, we use Theorem \ref{thmOne} to give a proof of the following theorem using only elementary geometric arguments.

\begin{theorem}\label{thmDimOne}
Let $(A, x, U)$ be any instance of the Orbit Problem such that all eigenvalues of $A\in \QQ^{d \times d}$ have distinct modulus up to complex conjugate pairs, $x \in \mathbb{Q}^d$, and $U$ a subspace of dimension at most one. Then it is decidable whether there exists an $n \in \mathbb{N}$ such that $A^nx \in U$.
\end{theorem}

The use of Theorem \ref{thmOne} reaches beyond the discrete-time Orbit Problem to Skolem's Problem as well (see e.g. \cite{ouaknine2015linear,halava2005skolem} for review), which is known to be an impenetrable problem where even simple cases escape our most advanced techniques. 
Indeed, Tao described the openness of this problem as ``faintly outrageous," as it indicates ``we do not know how to decide the halting problem even for linear automata" \cite{tao2008structure}.

Decidability of Skolem's Problem for linear recurrence sequences of dimension three and four was given in the 1980's \cite{shorey1984distance,vereshchagin1985occurrence}, along with decidability when there are at most three dominating characteristic roots of the sequence (see \cite{sha2019effective} for refinements of this work).
The methods used to prove such results rely crucially on sophisticated results
in transcendental number theory, particularly Baker's lower bounds on the magnitudes of linear forms in logarithms of algebraic numbers, and van der Poorten's results in the setting of $p$-adic valuations.
In contrast, we use Theorem \ref{thmOne} paired with a basic geometric insight to provide an alternate proof of decidability of Skolem's Problem for the case of arbitrarily many roots of largest modulus, so long as the set of such dominating roots contains a real root with largest multiplicity in the minimal polynomial.

\begin{theorem}\label{thmSkolemTwo}
Let $\{u_n\}_{n=0}^\infty$ be an order $d$ non-degenerate linear recurrence sequence, with $A$ the companion matrix, and suppose the initial terms form a vector $x \in \QQ^d$ non-trivial with respect to the matrix $A$. 
Suppose the linear recurrence sequence has $r\leq d$ distinct characteristic roots $\lambda_1,\dots,\lambda_r$ of largest modulus. Let $n_i$ denote the multiplicity of the root $\lambda_i$, $i=1,\dots,r$, in the minimal polynomial of $A$. Suppose $\lambda_1$ is real, and $n_1>n_i$, $i=2,\dots,r$.
Then it is decidable whether the sequence has a zero term.
\end{theorem}

We find the continuous version of the Orbit Problem to be amenable to our techniques as well.
This is because real projective space produces a rather forgiving environment for analyzing flows, and particularly for bounding flows into certain regions of space.
To this end, we give an elementary proof of decidability of the state of the art (unconditional) two-dimensional continuous-time Orbit Problem (Continuous Skolem-Pisot Problem) established in \cite{bell2010continuous}---in fact we prove a strengthened version of the statement given in \cite{bell2010continuous}.

\begin{theorem}\label{thmContTwo}
Let $dx(t)/dt = Ax(t)$, $A \in Gl(2, \mathbb{A})$, be an autonomous linear differential equation with $x(0) \in \mathbb{A}^{2}$, and let $U \subset \mathbb{A}^2$ be a line. If the eigenvalues of $A$ are real, it is decidable if the flow intersects $U$. If the eigenvalues of $A$ are complex, the flow intersects $U$ infinitely often.
\end{theorem}

We move to describe our techniques, in particular the basic geometric insight behind the results in this paper.

\subsection{Overview of Techniques}

Traditionally, work toward the Orbit Problem has leveraged heavy machinery from various domains of algebraic geometry, mathematical logic, and number theory---particularly algebraic and transcendental number theory. In contrast, our approach is to develop a machinery suitable for application to the Orbit Problem by leveraging basic geometric and dynamical insight.

In essence, we consider not the dynamical system $(\RR^d, A)$ the Orbit Problem is formulated in, but rather the associated induced dynamical system over real projective space (i.e. after projecting onto $S^{d-1}$ and identifying opposite points), with the same compactification performed in the continuous-time analog. Working instead in $\RR\PP^{d-1}$, we can ask equivalent versions of the Orbit Problem that are decidable if and only if the original formulation is decidable.

We find that by compactifying the state space in such a way, the asymptotic behavior of the system becomes far more understandable: in projective space, we have asymptotically stable (converging) orbits, while no-such behavior is expressed in Euclidean space.
As a consequence, we can use arguments unavailable in the original formulation. The central contribution of our paper can then be interpreted as a fine-grained examination of the $\omega$-limit sets---or more generally the stable manifolds---of dynamical systems $(\RR\PP^{d-1}, f)$ for $f \in PGL(\RR^d)$, namely linear maps and flows over real projective space, with this deepened understanding lending itself to deciding the discrete and continuous-time Orbit Problems.
In practice, working directly with real-projective space is not feasible, so instead we consider the evolution of the \emph{angles} between an orbit of a point and some subspace, which has the same desired effect. To make proofs fully rigorous and to speak of ``convergence" of orbits, we then translate between consideration of angles between a point and subspace, and the convergence of a sequence in projective space to a cycle or fixed point.

This reduction to projective space is both surprising and noteworthy for a number of reasons. First, we note that although we lose an entire manifold dimension after compactifying the space, we maintain all the necessary information for deciding the Orbit Problem---indeed, the reduction can be seen as keeping only the ``essential" aspects of the Orbit Problem. In addition, the compactification enables the application of general tools and intuition from the study of continuous dynamical systems over a compact space, for which there is far more structure and results than in the unbounded case, enabling the application of an entire area of mathematics not used to date.

In the setting of flows, asymptotic behavior of autonomous linear differential equations is dictated by the Lyapunov exponents. These exponents determine asymptotic behavior of flows by the magnitudes and signs of the \emph{real parts} of the eigenvalues.
When projecting flows in $\mathbb{R}^d$ onto $\mathbb{RP}^{d-1}$, all flows are asymptotic to stationary points or cycles, since the compact phase space $\mathbb{RP}^{d-1}$ precludes divergence of trajectories.
If a flow passes from one side to another of a target hyperplane then it must intersect it. By using this fact one may obtain stronger results than in the case of maps where asymptotic behaviors, even in real projective space, are far more subtle. We abuse this additional geometric structure present in the study of flows in $\mathbb{RP}^{d-1}$ when proving Theorem \ref{thmContTwo}.

To the best of our knowledge, no previous work on the Orbit Problem and related questions leverage the techniques introduced here. However, a few works come close in spirit to the asymptotic analysis used here, particularly work on termination of linear loops (see \cite{akshay2022robustness, ouaknine2014positivity, braverman2006termination, hosseiniIntegerTermination, tiwari2004termination}), where the observation that the largest eigenvalues dominate the behavior of the program is used repeatedly. 
A similar technique is well-known in the study of linear recurrence sequences, whereby restricting the spectral structure of the problem and using the closed-form solution of an LRS, it can be clear how the characteristic roots of maximal modulus dominate the asymptotic behavior.
Although this paper leverages the same basic principle, we take it much farther to obtain finer results.

\subsection{Related Work}

One of the original sources of motivation for studying the Orbit Problem was Skolem's Problem, due to the fact that Skolem's Problem is a special case of the Orbit Problem. Although Skolem's Problem is easy to describe, it is difficulty to prove deep theorems on account of a lack of machinery.
Indeed, in terms of lower bounds, it is known that Skolem's Problem is NP-hard \cite{blondel2002presence}. This translates to the Orbit Problem when restriction is not placed on the dimension of the target subspace. Moreover, deep work has shown that further advances on the Skolem and Orbit Problem for higher dimensions are likely to be hard due to the fact that such advances in related problems would entail major advances in Diophantine approximation \cite{ouaknine2014positivity}. Despite these difficulties, there continues to be notable advances toward Skolem's problem (see e.g. \cite{lipton2022skolem,biluMFCSskolem}), including recent promising new techniques by way of Universal Skolem Sets \cite{kenison2020skolem, luca2021universal, lucaMFCSuniversal, luca2023skolem}.

Another principle source of motivation in studying the Orbit Problem and related questions comes from the problem of program verification. Specifically, enormous effort has been dedicated to solving the ``Termination Problem," chiefly concerned with deciding whether a ``linear while loop" will terminate (see e.g. \cite{tiwari2004termination,braverman2006termination,hosseiniIntegerTermination,cook2006termination,cook2006terminator,bozga2012deciding,ben2013linear,bradley2005termination,ouaknine2014termination,karimov2022s,vyalyi2011orbits,luca2022algebraic}). 
In particular, we note that the Polytope Hitting Problem where the target is an intersection of half spaces, is closely connected to the higher-dimensional Orbit Problem, and more immediately translates to problems of program verification and termination of linear loops \cite{chonev2014polyhedron}.

\subsection*{Acknowledgements}
The author is grateful to David Cash for the continued feedback and support.
The author is supported by the National Science Foundation Graduate Research Fellowship Program under Grant No. 2140001.

\section{Preliminary results}\label{secPreliminaries}

The purpose of this section is to establish notation, as well as results and machinery used for the proofs of Theorems \ref{thmOne},\ref{thmSkolemTwo}, and \ref{thmContTwo} in Section \ref{sectionTheorems}.
Sections \ref{secAnglesBetweenFlats} and \ref{secAngleEvolution} are the richest, where the key lemmas of this paper are formed.
Specifically, we focus on development in the discrete-time setting for the discrete Orbit Problem. Once the results of this setting are established, techniques and intuition translate to the continuous setting in a very straight-forward way.

Given an instance $(A, x, U)$ of the Orbit Problem, always take $A \in \QQ^{d \times d}$, $x \in \QQ^d$, and $U$ presented via sets of basis vectors from $\QQ^d$. With such a basis in hand it is clearly decidable whether a vector $A^nx \in \QQ^d$ is in $U$.

For the purposes of this paper, there is no harm in additionally supposing that the instances of the Orbit Problem are non-degenerate. An instance $(A, x, U)$ of the Orbit Problem is \emph{degenerate} if there exists two distinct eigenvalues of $A$ whose quotient is a root of unity. If not, the instance is \emph{non-degenerate}. Any instance of the Orbit Problem can be reduced to a finite set of non-degenerate instances.

Let $\lambda_1,\dots,\lambda_m$ label the eigenvalues of $A \in \QQ^{d \times d}$, $m \leq d$. We always assume the roots are labelled so that $|\lambda_1|\geq |\lambda_2| \geq \cdots \geq |\lambda_m|$. In the case $|\lambda_1| = \cdots = |\lambda_r| > |\lambda_{r+1}|$ we say that $A$ has \emph{$r$ roots of maximal modulus}, or \emph{$r$ dominating roots}.

In this paper we work in ambient Euclidean space $\RR^d$ (before compactifying), with orbits $\{A^nx\}_{n=0}^\infty$ contained in $\QQ^d$. All standard algebraic operations, such as sums, products, root finding of polynomials, and computing Jordan canonical forms of rational matrices \cite{cai1994computing} are well-known to be computable, and hence we do not discuss the details of such procedures (see Section \ref{secAlgebraic} for details on the effective manipulation of algebraic numbers).

Recall every matrix $A \in \RR^{d \times d}$ can be written in the form $A = QJQ^{-1}$ where $Q$ is nonsingular and $J$ is the Jordan canonical form of $A$. Let $J=J(A)$ denote the Jordan matrix of $A$. For the purpose of this paper, we use the \emph{real Jordan canonical form}.

\begin{lemma}[Real Jordan Canonical Form]\label{lemJordan}
For any matrix $A \in \RR^{d \times d}$, there exists a basis in which the matrix of $A$ is a quasi-diagonal matrix $J = \text{Diag}(J_1,\dots,J_N)$ where each block $J_i$ is of form
\[
\begin{pmatrix}
\lambda_t & 1 & 0 &\cdots & 0 \\
0 & \lambda_t & 1 & \cdots & 0 \\
\vdots & & \ddots & \ddots \\
0 & 0 & \cdots& \lambda_t & 1 \\
0 & 0 & \cdots & & \lambda_t
\end{pmatrix}
\text{ or }
\begin{pmatrix}
\alpha_l & \beta_l & 1 & 0 & 0 & \cdots & & & & 0\\
-\beta_l & \alpha_l & 0 & 1 & 0 & \cdots & & & & 0\\
 0 & 0 & \alpha_l & \beta_l & 1 & 0 \\
 \vdots & \vdots  & -\beta_l & \alpha_l & 0 & 1 \\
 &  & & & \alpha_l & \beta_l \\
 &  & & &-\beta_l & \alpha_l \\
 & & & & & & \ddots \\
 & & & & & & \alpha_l & \beta_l & 1 & 0 \\
 \vdots & \vdots & & & & & -\beta_l & \alpha_l & 0 & 1 \\
 0& 0&\cdots & & & &  &  & \alpha_l & \beta_l\\
 0& 0& \cdots & & & &  &  & -\beta_l & \alpha_l
\end{pmatrix},
\]
where the $\lambda_t$, $t=1,\dots,u$ are the real eigenvalues of $A$, and the $\lambda_l, \bar{\lambda}_l = \alpha_l \pm i\beta_l$, $l=1,\dots,s$, $\beta>0$ are the complex eigenvalues of $A$. The sizes of the blocks are determined by the elementary divisors of $A$.
\end{lemma}

As such, we have $A = QJ Q^{-1}$, and if $A$ is algebraic (has algebraic entries), then $J$ and $Q$ are also algebraic matrices and their entries can be computed from the entries of $A$. The usefulness of the real Jordan canonical form follows from the fact that in the case the characteristic polynomial of $A$ has complex roots, the Jordan blocks of $J$ continue to have real entries. For further details on the real Jordan canonical form including proof of Lemma \ref{lemJordan} and effective methods for computing $J$, consult \cite{shilov2012linear}.

For the remainder of this paper, we suppose $J$ is the real Jordan canonical form of $A$. In addition, suppose Jordan blocks $J_i$, $i=1,\dots,N$ are of dimension $D_i$, and associated with eigenvalues $\lambda_i$. Of course, there may be many Jordan blocks associated with a single eigenvalue $\lambda_i$, so $N \geq m$, but we often abuse notation and say $J_i$ associates to $\lambda_i$ to speak of the collection of Jordan blocks associated with $\lambda_i$ --- in the event more precise language is needed we use it. Moreover, we remark that the largest Jordan blocks associated with an eigenvalue are determined by the multiplicity of the eigenvalue in the minimal polynomial of $J$ --- this is a useful fact used throughout our paper.

Moreover, we generally assume initial points $x$ are \emph{non-trivial} (with respect to a matrix $A$). That is, a point $x$ is non-trivial if in the Jordan basis of the matrix $A$, $x$ has at least one-non-zero component with respect to each Jordan block composing $J(A)$. This assumption prevents a ``collapse" to lower order instances, and is not technically necessary in most cases --- if it is removed then the statements are often very similar if not ultimately identical. Nonetheless, we keep the assumption as it simplifies matters.

We reduce the Orbit Problem to a simpler version where $A$ is taken to be a block-diagonal Jordan matrix by way of the following

\begin{lemma}\label{lemmaReduce}
Deciding instances $(A, x, U)$ of the Orbit Problem reduces to deciding instances $(J, x', U')$ where $J=J(A)$ is the Jordan matrix of $A$.
\end{lemma}
\begin{proof}
Let $A = QJQ^{-1}$ with $J$ the Jordan canonical form of $A$, $Q$ nonsingular, and recall $A^k = QJ^kQ^{-1}$. Then $A^k x \in U$ for some $k \in \mathbb{N}$ implies
\[
QJ^kQ^{-1} x \in U, \text{ implying } J^k(Q^{-1}x) \in Q^{-1}U.
\]
Letting $Q^{-1}x = x'$ and $Q^{-1}U = U'$, it follows that $A^kx \in U$ if and only if $J^kx' \in U'$.
\end{proof}

As a consequence of Lemma \ref{lemmaReduce}, for the remainder of this paper when working with an instance $(A, x, U)$ of the Orbit Problem, we implicitly work with the equivalent instance $(J, x', U')$.

\subsection{Algebraic numbers and computing eigenvalues, eigenvectors}\label{secAlgebraic}

Let $\AA$ denote the field of algebraic numbers; a complex number $\alpha$ is \emph{algebraic} if it is a root of a single variable polynomial with integer coefficients. The matrix $A$ is taken to have rational entries, so all eigenvalues $\lambda$ of $A$ are in $\AA$. This paper requires we effectively perform various operations with the eigenvalues of $A$, which may not lie in $\QQ$. Fortunately, there is a sizable literature concerning computation with algebraic numbers (see \cite{AGAlgorithms,cohen2013course}, or the appendix of \cite{chonev2016complexity} for a useful review). Below, we summarize the basic properties of the effective manipulation of algebraic numbers pertaining to this paper.

For $\alpha \in \AA$, the \emph{defining polynomial} of $\alpha$, denoted $p_\alpha$, is the unique polynomial of least degree vanishing at $\alpha$, where the coefficients do not have common factors. Define the \emph{degree} and \emph{height} $H(p)$ of $\alpha$ to be the degree of $p$, and the maximum of the absolute values of $p$'s coefficients, respectively.
As such, standard finite encoding of an algebraic number $\alpha$ is as a tuple composed of its defining polynomial, and rational approximations of its real and imaginary parts of precision sufficient to distinguish $\alpha$ from the other roots of $p_\alpha$. That is, $\alpha \in \AA$ takes the representation $(p_\alpha, a, b, r) \in \ZZ[x] \times \QQ^3$, when $\alpha$ is the unique root of $p_\alpha$ inside a circle in $\CC$ of radius $r$, centered at $a+bi$. 
This representation of $\alpha$ is well-defined and computable in part thanks to a useful separation bound of Mignotte \cite{mignotte1982some}, which asserts that for distinct roots $\alpha, \beta$ of polynomial $p \in \ZZ[x]$
\[
|\alpha - \beta| > \frac{\sqrt{6}}{d^{(d+1)/2}H^{d-1}},
\]
where $d = deg(p)$ and $H = H(p)$. As a consequence, when $r$ is less than the root separation bound, we have equality checking between algebraic numbers. Given distinct $\alpha, \beta \in \AA$ these are roots of $p_\alpha p_\beta$, which are of degree at most $deg(\alpha)+deg(\beta)$, and of height at most $H(\alpha)H(\beta)$. Then one may compute $\alpha+\beta$, $\alpha\beta$, $1/\alpha$, $\overline{\alpha}$, $|\alpha|$, decide whether $\alpha > \beta$ for distinct algebraic $\alpha$ and $\beta$, and more.

For the remainder of this paper, we do not make explicit use of these methods, implicitly assuming they are used in the following procedures for deciding instances of the higher-dimensional Orbit Problem. Indeed, it should be clear from the algorithms given in the sequel that the we never leave the field of algebraic numbers. Although, the techniques of this paper do not in fact require the precision the effective manipulation of algebraic numbers provides: all our techniques are quite amenable to the finite-precision setting. So long as precision can be increased as needed, for our purposes a finite approximation of a number paired with an error bound more than suffices. Nonetheless, for the sake of clarity we opt to compute with algebraic numbers as discussed above.

As a consequence of these observations, we see that we have effective techniques for computing the eigenvalues and eigenvectors of rational matrices $A$.

\subsection{Linear recurrence sequences}\label{secLinearRecurrence}

We now consider basic aspects of linear recurrence sequences. For a rich treatment see \cite{everest2003recurrence}. A \emph{linear recurrence sequence (LRS)} over $\QQ$ is an infinite sequence $\{u_n\}_{n=0}^\infty$ of terms in $\QQ$ satisfying the recurrence relation
\[
u_{n+d} = a_{d-1}u_{n+d-1} + \cdots + a_0 u_n
\]
where $a_0,\dots,a_{d-1} \in \QQ$, with $a_0 \neq 0$ and $u_j \neq 0$ for at least one $j$ in the range $0 \leq j \leq d-1$. We say such an LRS has \emph{order $d$}. We call the $a_0,\dots,a_{d-1}$ the \emph{coefficients} of the sequence $\{u_n\}$ and the \emph{initial terms} of $\{u_n\}$ are $u_0,\dots,u_{d-1}$. The \emph{characteristic polynomial} of the sequence $\{u_n\}$ is
\[
x^d - a_{d-1}x^{d-1} - \cdots - a_0 = \prod_{i=1}^k(x - \lambda_i)^{m_i},
\]
with the $\lambda_1,\dots,\lambda_k$ called the \emph{characteristic roots} of the sequence  $\{u_n\}$. 

We call the sequence \emph{simple} if $k=d$, so $m_1 = \cdots =m_k = 1$, and \emph{non-degenerate} if $\lambda_i/\lambda_j$ is not a root of unity for any distinct $i, j$. The study of arbitrary LRS can be reduced effectively to that of non-degenerate LRS by partitioning the original LRS into finitely many non-degenerate subsequences.

Given an LRS, define
\[
A =
\begin{pmatrix}
a_{d-1} & a_{d-2} & \cdots & a_{1} & a_0 \\
1 & 0 & \cdots & 0 & 0\\
0 & 1 & \cdots & 0 & 0 \\
\vdots & \vdots & \ddots & \vdots & \vdots \\
0 & 0 & \cdots & 1 & 0
\end{pmatrix}
\]
to be the companion matrix of the LRS, and take the vector $x \in \QQ^{d}$ as $(u_{d-1}, \dots, u_0)^T$ of initial terms of the LRS. Then iteration of $A$ over $x$ acts as a ``shift" on the entries of $x$, shifting in the next term of the LRS and dropping the oldest term.
In addition, the characteristic polynomial of the LRS is the characteristic polynomial of $A$, and the characteristic roots of the LRS are the eigenvalues of $A$.

\subsection{Angles between flats}\label{secAnglesBetweenFlats}

The algorithm presented in this paper requires the computation of angles between Euclidean subspaces. The subject of computing angles between Euclidean subspaces --- also known as computing angles between \emph{flats} --- is a well-studied area and there are many effective and efficient techniques for doing so. See \cite{BFb0062107,gunawan2005formula,bjorck1973numerical,jiang1996angles} for readable introductions to the subject, and a variety of effective techniques for computing angles between subspaces.

Let $U, W \subset \RR^d$ be $k$ and $l$ dimensional subspaces, respectively, presented by orthogonal rational basis $\{u_1,\dots,u_k\}$ and $\{w_1,\dots,w_l\}$. Note any linearly independent set of rational vectors can be effectively transformed into an orthogonal set of rational vectors via the Gram-Schmidt Process (before normalization). If instead it is desired that the basis be orthonormal, then the entries of the vectors will be algebraic after the renormalization step.

In this paper we are primarily concerned with determining the minimal angle $0 \leq \theta_1 \leq \pi/2$ between two subspaces, defined as the least angle between any pair of non-zero vectors from $U$ and $W$. This is the \emph{first principal angle} between subspaces, whose variational characterization is

\begin{equation}\label{principleAngle}
\theta = \min\left\{ \arccos\left( \frac{|\langle u, w\rangle|}{\|u\|~ \|w\|}\right) : u \in U, w \in W, \text{ and } u, w \neq 0\right\},
\end{equation}

where $\langle \cdot,\cdot\rangle$ is the standard inner (dot) product. We refer to the quantity expressed above in Equation \ref{principleAngle} as the \emph{minimal angle between subspaces $U$ and $W$}. We always assume $U$ and $W$ have trivial intersection $U \cap W = \{0\}$, so that the minimal angle is well-defined and non-zero.

It is well-known that the variational characterization of singular values implies a variational characterization of the angles between subspaces \cite{golub2013matrix}. Indeed, as a consequence of the variational characterization of the singular value decomposition of a $m \times n$ matrix $A$, we see that
for $x \in \CC^m$, $y \in \CC^n$, then
\[
\sigma_1 = \max_{x \in \CC^m, y \in \CC^n} \frac{|x^*Ay|}{\|x\|\ \|y\|},
\]
where $x^*$ is the conjugate transpose and $\sigma_1$ is the first singular value of $A$.
Dropping the $\arccos$ term in Equation \ref{principleAngle}, we obtain a new problem, namely maximizing the quantity
\begin{equation}\label{eqMaximize}
\sigma = \max\left\{ \frac{|\langle u, w\rangle|}{\|u\|\ \|w\|} : u \in U, w \in W\right\}.
\end{equation}
Note then that $\sigma \leq 1$. The pair of vectors $u, v$ maximizing the above quantity can be used to obtain the minimal angle between the subspaces $U$ and $W$.

Combining the above observations and restricting to the reals, take matrices $X \in \mathbb{R}^{n \times k}$ and $Y \in \mathbb{R}^{n \times l}$ to have columns consisting of orthonormal bases $\{u_1,\dots,u_k\}$ and $\{w_1,\dots,w_l\}$ of $U$ and $W$, respectively. The optimization problem in Equation \ref{eqMaximize} can then be written as 
\[
\sigma = \max_{x \in \mathbb{R}^k, y \in \mathbb{R}^l} \frac{x^{\top} X^{\top} Y y}{\|x\|\ \|y\|},
\]
the solution to which is the largest singular value of $X^{\top} Y$, by the variational characterization of the singular value decomposition.
For more details and proof of the fact that cosines of principal angles come from the singular value decomposition see \cite{knyazev2002principal}, or refer to Chapter 5 of \cite{meyer2023matrix} for more details on the variational characterization of singular value decomposition, and additional effective methods for computing the minimal angle between subspaces.

The value $\theta = \arccos \sigma$ is the minimal angle between subspaces $U, W$. However, we need not learn $\theta$ --- in fact $\theta$ need not be an algebraic number due to the $\arccos$. Rather, it is sufficient for the purposes of this paper to instead work only with the cosine of the minimal angle, namely $\sigma$. The algorithm presented in this paper only requires that we compare the minimal angle between subspaces, which we can accomplish equipped just with $\sigma$: supposing $\sigma'<\sigma$, we know $\arccos(\sigma') > \arccos(\sigma)$, and hence the angle corresponding to $\sigma$ is smaller than the angle corresponding to $\sigma'$. In fact it  is sufficient not to work with $\sigma$, but $\sigma^2$, which is the largest of  the positive eigenvalues obtained in the singular value decomposition.

The singular value $\sigma$ is merely the square root of an eigenvalue computed in the singular value decomposition (or just the eigenvalue in the $\sigma^2$ case), and is hence algebraic, admitting a finite representation as per Section \ref{secAlgebraic}. As a consequence, we may effectively compare the minimal angles between different pairs of Euclidean subspaces described by rational basis vectors.

For the remainder of this paper, when we speak of angles between subspaces, we are implicitly speaking of the singular value $\sigma$ expressed above. Define the function
\begin{equation}\label{eqGamma}
\Gamma(U, W) = \max_{x \in \mathbb{R}^k, y \in \mathbb{R}^l} \frac{x^{\top} X^{\top} Y y}{\|x\|\ \|y\|} = \sigma,
\end{equation}
to be the function computing the minimal angle between subspaces $U$ and $W$ as above.
If one or both of $U$ and $W$ are one-dimensional subspaces, i.e. $U = \text{span}(x)$ for $x \in \mathbb{Q}^d$, then $\Gamma(x, W)$ is understood to mean $\Gamma(U, W)$ to allow for more natural notation.

Whether $\Gamma(U, W)$ returns $\sigma$ or $\sigma^2$ does not matter for the purposes of this paper --- intuitively $\Gamma(U, W)$ computes the minimal angle between the subspaces ($\arccos \Gamma$ is the actual minimal cosine angle), and everything remains computable: the columns of $U$ and $W$ will always be composed of rational (or possibly algebraic) numbers, and $\Gamma$ is effectively computable with such input.

\subsection{The Orbit Problem in real projective space}\label{secAngleEvolution}

Real projective space $\RR\PP^{d-1}$ is defined to be the quotient of $\RR^{d}\setminus \{0\}$ by the equivalence relation $x \sim \alpha x$ where $\alpha \neq 0$ is real and $x \in \RR^d$. Clearly, it suffices to consider only vectors $x$ with Euclidean norm $\|x \| = 1$. Hence, geometrically, projective space is the space of lines through the origin, or alternatively the space obtained by identifying opposite points on the unit sphere $S^{d-1}$. The projective linear group $PGL(\RR^d) = GL(\RR^d)/\alpha\cdot\text{Id}$, $\alpha \neq 0$, is the induced action of the general linear group on projective space. The induced action of non-invertible square matrices over elements of projective space is similarly defined.

Write $\PP:\RR^d\setminus \{0\} \rightarrow \RR\PP^{d-1}$ for the projection of elements from Euclidean space to real projective space, denoting elements of $\RR\PP^{d-1}$ by $p = \PP x$, where $0 \neq x \in \RR^d$. 
A metric on $\RR\PP^{d-1}$ is given by
\[
d(\PP x, \PP y) = \min\left(\left|\left| \frac{x}{\|x\|} - \frac{y}{\|y\|}\right|\right|,
\left|\left|\frac{x}{\|x\|} - \frac{-y}{\|y\|} \right|\right| \right).
\]
In the case that $W$ is a subspace of $\RR^d$, define the distance between a point $x$ and the subspace to be
\[
d(x, W \cap S^{d-1}) = \inf_{y \in W \cap S^{d-1}}\|x-y\| = \min_{y \in W} d(\PP x, \PP y) = d(\PP x, \PP W).
\]
Observe that these metrics are intimately related to the function $\Gamma$ defined in the previous subsection, as expressed in the following trivial, albiet useful lemma.

\begin{lemma}\label{lemAngleConverge}
Let $U$ be a subspace of $\QQ^d$, $x \in \QQ^d$, and $A \in \QQ^{d \times d}$.
Suppose $\Gamma(A^nx, U)\rightarrow 1$ as $n \rightarrow \infty$. Then
\[
\lim_{n\rightarrow \infty}d(\PP A^nx, \PP U) = 0.
\]
\end{lemma}
\begin{proof}
We show the lemma for the case that $U$ is a one-dimensional subspace spanned by $y \in \QQ^d$, which is easily generalized to higher-dimensional $U$. 

We have
\[
\Gamma(A^nx, U) = \frac{|A^nx \cdot y|}{\|A^nx\|\ \|y\|}, \ n \in \NN.
\]
Then clearly if
\[
\frac{|A^nx \cdot y|}{\|A^nx\|\ \|y\|} \rightarrow 1 \text{ as } n \rightarrow \infty,
\]
we have
\[
d(\PP A^nx, \PP y) = \min\left(\left|\left| \frac{A^nx}{\|A^nx\|} - \frac{y}{\|y\|}\right|\right|,
\left|\left|\frac{A^nx}{\|A^nx\|} - \frac{-y}{\|y\|} \right|\right| \right) \rightarrow 0 \text{ as } n \rightarrow \infty.
\]
This readily generalizes to the case when $U$ is of higher-dimension, as we sketch below.

The function $\Gamma$ returns the cosine of the minimal angle between two subspaces. Hence, recalling the variational characterization of the minimal angle in Equation \ref{principleAngle} after dropping the $\arccos$, we have
\[
\Gamma(U, W) = \sigma = \max\left\{\left( \frac{|\langle u, w\rangle|}{\|u\|~ \|w\|}\right) : u \in U, w \in W\right\}.
\]
So, abusing notation and letting $A^nx$ label the subspace spanned by $A^nx$, we have
\[
\Gamma(A^nx, W) = \sigma = \max\left\{\left( \frac{|\langle A^nx, w\rangle|}{\|A^nx\|~ \|w\|}\right) : w \in W\right\}.
\] 
Combining this with the fact that
\[
d(\PP A^nx, \PP W) = \inf_{y \in W \cap S^{d-1}}\|A^nx-y\| = \min_{y \in W} d(\PP A^nx, \PP y),
\]
the general statement is immediate upon considering the proof of the case when $U$ is one-dimensional, as shown above.
\end{proof}

Given $x \in \RR^d$, let $X$ denote the one-dimensional subspace (line) spanned by $x$. Then, for a matrix $A \in \RR^{d \times d}$ and subspace $U$, we see that for every $n \in \NN$, $A^nx \in U$ if and only if $A^nX \subseteq U$, where $A^nX$ is the subspace spanned by the vector $A^nx$. Hence, given an instance $(A, x, U)$ of the Orbit Problem, it is immediate that there is an $n \in \NN$ such that $A^nx \in U$ if and only if $\PP A^nx \in \PP U$.

These considerations lead us to the simple yet powerful conclusion that \emph{by projecting the Orbit Problem into projective space we still retain the needed information to decide the Orbit Problem --- indeed the formulations are equivalent --- and yet in projective space the asymptotic behavior of orbits is far more manageable than in Euclidean space.}  In projective space, we have asymptotically stable (converging) orbits, while no-such behavior is expressed in Euclidean space. We abuse this structure in the sequel, applying the basic insight that if an orbit of a dynamical system monotonically converges to some set, and the attracting set is separated from the target set by some $\epsilon > 0$, then by running the system for finite time we can decide whether the orbit intersects the target set. See \cite{colonius2014dynamical} for details on linear dynamical systems over real projective space.

\subsection{Convergence of orbits toward attracting subspaces in $\mathbb{RP}^{d-1}$}

In this section we establish the lemmas underlying Theorem \ref{thmOne}, and hence all the other results in this paper.
The following lemmas and the arguments used are technical, which perhaps obfuscates the underlying intuition and central concepts which are in-fact quite simple.
To this end, we take a moment to present the intuition behind the following results, which should aid in clarifying one's understanding.

The central concept is this. Let
\[
\Lambda = \text{Diag}(x_1,x_2,\dots,x_d)
\]
be a diagonal matrix with $x_1>x_2>\cdots >x_d\geq 1$. If $y \in \mathbb{Q}^d$, then iteration of $\Lambda$ over $y$ sends the first entry of $y$ to $\infty$ faster than the others since $x_1 > x_i$, $i\neq 1$, i.e.
\[
\lim_{n\rightarrow \infty}\frac{x_i^n}{x_1^n} = 0 \text{ for } i\neq 1.
\]
Hence, $\Gamma(\Lambda^ny, e_1)\rightarrow 1$ as $n\rightarrow \infty$ where $e_1$ is the first standard basis vector. In particular, following Lemma \ref{lemAngleConverge},
\[
\lim_{n\rightarrow \infty}d(\PP \Lambda^nx, \PP e_1) = 0.
\]
This observation underlies the idea behind Theorem \ref{thmOne}: if the orbit of a point in projective space approaches an attracting subspace that does not intersect with the target subspace, then eventually the orbit must be bound away form the target subspace for all time, so we have only a finite number of iterations for which intersection must be checked.

The following lemmas deal with a generalization of the above example, where $\Lambda$ is not a real diagonal matrix but rather a matrix in Jordan normal form. However, one may just as easily abstractly treat the $x_i$ composing $\Lambda$ not as real numbers, but rather as Jordan blocks. Then the same principle holds: a certain Jordan block or collection of identical Jordan blocks will ``dominate" the asymptotics, and hence $\Gamma(\Lambda^ny, E)\rightarrow 1$ as $n\rightarrow \infty$. The critical difference being that $E$ may not be one-dimensional, but is a direct sum of the one or two-dimensional subspaces invariant under each real or complex Jordan block dominating the asymptotics, respectively.

Equipped with this intuition, we begin by stating the following well-known result about Jordan blocks, easily proved by induction.

\begin{lemma}[Powers of Jordan blocks]\label{lemJordanPowers}
For a matrix $K = \text{Diag}(J_1,\dots,J_N)$ in Jordan canonical form, its $n$th power is given by $J^n = \text{Diag}(J_1^n,\dots,J_N^n)$. Where, for real eigenvalues $\lambda_i$ we have
\[
J_i^n = 
\begin{pmatrix}
\lambda_i^n & n\lambda_i^{n-1} & {n \choose 2}\lambda_i^{n-2} & \cdots & {n \choose D_i-1} \lambda_i^{n-(D_i-1)} \\
0 & \lambda_i^n & n\lambda_i^{n-1} & \cdots & {n \choose D_i-2} \lambda_i^{n-(D_i-2)} \\
\vdots & & \ddots & \ddots & \vdots \\
0 & 0 & \cdots & \lambda_i^n & n\lambda_i^{n-1} \\
0 & 0 & \cdots & 0 & \lambda_i^n
\end{pmatrix}
\]
where $D_i$ is the dimension of the block $J_i$, and ${n \choose k} = 0$ if $n < k$. And when $J_i$ is a Jordan block associated with a complex conjugate eigenvalue pair expressed in the real canonical form as in Lemma \ref{lemJordan}, letting
\[
R_i = \begin{pmatrix}
\alpha_i & \beta_i \\
-\beta_i & \alpha_i
\end{pmatrix},
\]
we have
\[
J_i^n = 
\begin{pmatrix}
R_i^n & nR_i^{n-1} & {n \choose 2}R_i^{n-2} & \cdots & {n \choose D_i-1} R_i^{n-(D_i-1)} \\
0 & R_i^n & nR_i^{n-1} & \cdots & {n \choose D_i-2} R_i^{n-(D_i-2)} \\
\vdots & & \ddots & \ddots & \vdots \\
0 & 0 & \cdots & R_i^n & nR_i^{n-1} \\
0 & 0 & \cdots & 0 & R_i^n
\end{pmatrix}
\]
with $J_i$ of dimension $2D_i$.
\end{lemma}

Note that that when $\lambda = \alpha \pm i\beta$ is complex $|\lambda| = \sqrt{\alpha^2 + \beta^2}$, hence with $\theta \in [0, 2\pi)$ determined by $\cos \theta = \alpha / \sqrt{\alpha^2 + \beta^2}$, we can rewrite the matrix $R$ above as
\[
|\lambda|R \text{ with } R = R_i = \begin{pmatrix}
\cos\theta & \sin\theta \\
-\sin\theta & \cos\theta
\end{pmatrix}.
\]
Hence $R$ describes a rotation by the angle $\theta$ followed by multiplication by $|\lambda|$. Letting $H$ denote the $2D_i$ dimensional nilpotent matrix of ones just above the principal diagonal, for Jordan block $J_i$ associated with a complex conjugate eigenvalue pair, we have
\begin{equation}\label{complexJordanIterate}
J_i^nx = \sum_{i=0}^{D_i-1}{n \choose i}|\lambda|^{n-i}\tilde{R}^{n-i}H^i x
\end{equation}
where $\tilde{R}$ is the block diagonal matrix with blocks $R$, and $n\geq D_i-2$.

The following lemmas establish the asymptotic behavior of a non-zero vector under iteration of Jordan blocks. Tiwari and Braverman in \cite{tiwari2004termination,braverman2006termination} touch upon a similar idea,  although the outcome and use is different.

We also note that the following Lemmas \ref{lemJordanBlockLimitOne}, \ref{lemJordanBlockLimitTwo}, and \ref{lemJordanBlockLimitThree} are essentially folklore in dynamical systems; the asymptotic behavior of real projective transformations is well understood --- see e.g. \cite{colonius2014dynamical,ayala2014topological,KUIPER197613,he2017topological} which all contain relevant background and results related to the Lemmas below. However, we write the following Lemmas in a way that corresponds their computable nature, thereby translating them into a form suitable for the problem this paper is attacking.

\begin{lemma}\label{lemJordanBlockLimitOne}
Let $J_i$ label a Jordan block of form
\[
J_i = \begin{pmatrix}
D & I & 0 &\cdots & 0 \\
0 & D & I & \cdots & 0 \\
\vdots & & \ddots & \ddots \\
0 & 0 & \cdots& D & I \\
0 & 0 & \cdots & & D
\end{pmatrix}
\]
with respect to eigenvalue $\lambda_i$, $|\lambda_i| > 0$, where $D = \lambda_i$ and $I = 1$ if $\lambda_i$ is real. When $\lambda_i,\overline{\lambda_i}$ are complex set
$D = \begin{pmatrix}
\alpha_i & \beta_i \\
-\beta_i & \alpha_i
\end{pmatrix}$,
and let $I$ be the two-by-two identity matrix.

Take $J_i$ to have dimension $\delta$, and restrict its action to $\QQ^{\delta}$, letting $x \neq 0 \in \QQ^{\delta}$.
Let $P$ denote the one (two) dimensional subspace invariant under $J_i$ when $\lambda_i$ is real (complex).
Then
\[
\lim_{n\rightarrow \infty} d(\PP J_i^nx, \PP P) = 0,
\]
where $d$ is the metric for real projective space defined previously.
\end{lemma}
\begin{proof}
By Lemma \ref{lemAngleConverge}, it suffices to show that $\Gamma(J_i^n x, P) \rightarrow 1$ as $n \rightarrow \infty$.
Recall that $\arccos(\Gamma(J_i^n x, P))$ is the (minimal) cosine angle between subspaces $J_i^nx$ and $P$.
Using Lemma \ref{lemJordanPowers} we have
\[
J_i^n = \begin{pmatrix}
D^n & nD^{n-1} & {n \choose 2}D^{n-2} & \cdots & {n \choose \delta-1} D^{n-(\delta-1)} \\
0 & D^n & nD^{n-1} & \cdots & {n \choose \delta-2} D^{n-(\delta-2)} \\
\vdots & & \ddots & \ddots & \vdots \\
0 & 0 & \cdots & D^n & nD^{n-1} \\
0 & 0 & \cdots & 0 & D^n
\end{pmatrix}
\]
and thus
\[
J_i^nx
=
\begin{pmatrix}
D^nx_1 + \cdots + {n \choose \delta-1} D^{n-(\delta-1)}x_\delta\\
D^nx_2 + \cdots + {n \choose \delta-2} D^{n-(\delta-2)}x_\delta\\
\vdots \\
D^nx_{\delta-1} + nD^{n-1}x_\delta\\
D^nx_\delta
\end{pmatrix},
\]
where $x = (x_1,\dots,x_\delta)^\top$, and the components $x_i$ are either $1\times 1$ or $2 \times 1$ depending on whether $D$ is $1 \times 1$ or $2 \times 2$.

Dividing the components of $J_i^nx$ by $(J_i^nx)_1$ (the first component), and taking a limit (element-wise) we see that
\[
\lim_{n\rightarrow \infty} \left|\frac{J_i^n(x)}{(J_i^nx)_1}\right| = (1,0,\dots,0)^\top,
\]
as a consequence of the fact that the first term in the vector (corresponding to the subspace $P$) has the highest polynomial growth. This limit also indicates that when $D$ is $2 \times 2$, the vector approaches a two-dimensional subspace.

It is then immediate that
\[
\Gamma(J_i^n x, P) \rightarrow 1 \text{ as } n \rightarrow \infty,
\]
from the definition of $\Gamma$, and the statement follows.
\end{proof}

\begin{remark}
The statement of Lemma \ref{lemJordanBlockLimitOne} did not require $|\lambda_i|>1$, because it holds if $|\lambda_i|=1$ or $|\lambda_i| < 1$: when $|\lambda_i|=1$, the rate of growth of the first term of the vector is still higher than the others, and when $|\lambda_i| < 1$, the rate of decline of the first component of the vector is slowest, and hence the statement still holds. Furthermore, the statement continues to hold independent of whether $\lambda_i$ is positive or negative --- in all cases the angle between the vector and the one or two-dimensional space invariant under $J_i$ approaches zero. Finally, we remark that so long as $x$ has at least one non-zero component the lemma holds: after a few iterations all the first component of $J_i^nx$ will be non-zero and the asymptotics kick in accordingly.
\end{remark}

With Lemma \ref{lemJordanBlockLimitOne} in hand, we obtain the following set of lemmas giving asymptotic behavior of orbits when there are multiple Jordan blocks.

The following lemma states that when there are many Jordan blocks all associated with eigenvectors of the same modulus, the largest block(s) dictate the asymptotic behavior of orbits.

\begin{lemma}\label{lemJordanBlockLimitTwo}
Let $J = \text{Diag}(J_1,\dots,J_N)$ be a block diagonal matrix of Jordan blocks, where each $J_i$ is associated with a real or complex eigenvalues $\lambda_i$, such that $|\lambda_1| = |\lambda_2| = \cdots = |\lambda_N|$. Let $D_i$ denote the dimension of $J_i$ when $\lambda_i$ is real, and half the dimension of $J_i$ when $\lambda_i$ is complex ($J_i$ is the block associated to a complex conjugate pair). Suppose $D_1 = \cdots =D_T > D_{T+1} \geq \cdots \geq D_N$, with $T \leq N$.

Then there is a subspace $P$ of dimension $k$, $1 \leq k \leq 2T$, determined by the eigenvalues and Jordan blocks, such that
\[
\lim_{n\rightarrow \infty} d(\PP J^nx, \PP P) = 0
\]
for any $x$ with at least one non-zero component with respect to the Jordan blocks comprising $J$.
\end{lemma}
\begin{proof}
It suffices to prove the statement by analyzing two cases: the first case is when $T = 1$, and the second case is when $T =N$. Once these cases have been shown they can be easily combined to obtain the general statement.

\smallskip
\emph{Case 1: T = 1.}  
When $T = 1$ we have a single Jordan block $J_1$ such that $D_1 > D_2 \geq \cdots \geq D_N$.
Hence, the iteration of $J$ induces the highest polynomial growth in the components corresponding to the one or two-dimensional invariant subspace under the first Jordan block $J_1$, cf. the solution formula Lemma \ref{lemJordanPowers} above.

Following the proof of Lemma \ref{lemJordanBlockLimitOne}, it is clear that the component-wise limit of $|J^nx / (J^nx)_1|$ as $n \rightarrow \infty$ gives $(1,0,\dots,0)^\top$.
Thus, all elements $\PP x$ where $x$ has at least one non-zero component with respect to $J_1$, get taken to $\PP P$ in the induced dynamical system in real projective space, where $P$ is either one or two-dimensional depending on whether $J_1$ is associated with a real or complex-conjugate pair of eigenvalues, cf. Lemma \ref{lemJordanBlockLimitOne}.

\smallskip
\emph{Case 2: T = N.}
In the case, all the $D_i$ are equal, and hence no subset of larger Jordan blocks composing $J$ dictate the asymptotic behavior. Let $P_i$ denote the one or two-dimensional invariant subspace under $J_i$, depending on whether $\lambda_i$ is real or complex, respectively.
If $x_{J_i}$ denotes the components of $x$ with respect to block $J_i$, then we know by Lemma \ref{lemJordanBlockLimitOne} that iteration of $J_i$ over $x_{J_i}$ brings $\PP x_{J_i}$ to $\PP P_i$ as $n\rightarrow \infty$.

We first note that adding many Jordan blocks of the same dimension, all corresponding to the \emph{same eigenvalue} does not increase the dimension of $P=P(x, J)$. This follows from the argument used in proving Lemma \ref{lemJordanBlockLimitOne}. Namely, suppose $J$ is a block diagonal matrix composed of $N$ repeated Jordan blocks $J_i$, and let $x$ be some vector such that $x_{J_i} \neq 0$ for each $J_i$. Then take the limit of $|J^{n}x / (J^{n}x)_j|$ as $n\rightarrow \infty$, where $|(J^nx)_j|$ is the largest component of $J^nx$. Then the limiting vector will have form $(c_1,0,\dots,c_i,0,\dots,1,0,\dots,c_N,0,\dots0)^\top$, where $0 \leq c_i \leq 1$, and the $1$ is in the $j$th position, and the values of the $c_i$ can be effectively computed with the coordinates of $x$ through taking appropriate ratios of the components of $x$.
Hence, when the repeated Jordan blocks all correspond to a real eigenvalue, the invariant subspace $P$ is of dimension one. And when the repeated Jordan blocks correspond to a complex conjugate pair, $P$ is of dimension two.

Hence provided with this fact, in the $T=N$ case we may assume without any loss of generality that each Jordan block $J_i$ (and therefore invariant space $P_i$) is paired with a distinct eigenvalue or complex conjugate pair (up to a $-1$ factor), with no Jordan block repeated more than once. Then, as an immediate consequence of the block-diagonal structure of $J$, we have that
\[
P = \bigoplus_i P_i,
\]
where the $P_i$ are the one or two-dimensional subspaces invariant under corresponding collections of Jordan blocks associated with the same eigenvalue(s), and $\oplus$ is the direct sum.
Using Lemma \ref{lemJordanBlockLimitOne}, we have
\[
\lim_{n\rightarrow \infty} d(\PP J^nx, \PP P) = 0.
\]
In addition, following from the above argument it is clear that $\text{dim}(P) = \sum_i \text{dim}(P_i)$, and hence $\text{dim}(P) = k$ with $1 \leq k \leq N=T$.

\smallskip
Collecting the arguments establishing Cases 1 and 2 above, the statement follows: the subspace the dynamics are asymptotic to is determined by the Jordan blocks of largest dimension.
\end{proof}

\begin{remark}
In Case 2 of the proof of Lemma \ref{lemJordanBlockLimitTwo}, we emphasize that as long as $x$ has one non-zero component with respect to each Jordan block then the proof goes through. What our proof perhaps leaves unclear in regard to this matter is that in the limit discussed in the second paragraph, we have $0 \leq c_i \leq 1$, and hence it is possible that a $c_i=0$. But this is fully allowed, since in this case the particular orbit will still approach $P$, even though it is simply approaching a particular subset of $P$.
\end{remark}

\begin{remark}
We also note that in Case 2 of the proof of Lemma \ref{lemJordanBlockLimitTwo}, we underscore the fact that when there are many Jordan blocks of the same size but that do not correspond to a repeated eigenvalue, but rather the same eigenvalue up to a factor of $-1$, then the orbit $\{\PP J^{2n}x\}$ will converge to one subset $\PP P$ in projective space, while the orbit $\{\PP J^{2n+1}x\}$ will converge to another subset $\PP P'$ of projective space. This follows from the fact that certain components will change sign with every iteration.
However, for the purpose this paper and the proofs of Theorems \ref{thmOne} and \ref{thmSkolemTwo}, we assume non-degeneracy so no two distinct eigenvalues of $A$ will be equal up to a $-1$ factor, and hence this aspect can be safely ignored.
\end{remark}

The following lemma establishes the case when the Jordan blocks composing a block-diagonal Jordan matrix $J$ correspond to eigenvalues of different magnitude.

\begin{lemma}\label{lemJordanBlockLimitThree}
Let $J = \text{Diag}(J_1,\dots,J_N)$ be a block diagonal matrix, where each Jordan block $J_i$ is associated with a distinct real eigenvalue or complex conjugate pair.
Moreover, suppose $|\lambda_1| > |\lambda_2| > \cdots > |\lambda_N|$.
Then there is a subspace $P$ of dimension one (two) whenever $\lambda_1$ is real (complex) invariant under $J_1$, such that
\[
\lim_{n\rightarrow \infty} d(\PP J^n x, \PP P) = 0,
\]
where $d$ is the metric over real projective space, and $x$ has at least one non-zero component with respect to block $J_1$.
\end{lemma}
\begin{proof}
It suffices to show that the Jordan block $J_1$ associated with the largest eigenvalue dominates the dynamics asymptotically, independent of the sizes of the $J_i$.

That $J_1$ dominates the dynamics in the limit is a trivial consequence of the fact that
\[
\lim_{n\rightarrow \infty}\frac{\sum_{i=0}^{l} p_i(n)|\mu|^{n-i}}{|\lambda|^n} = 0, \text{ whenever } |\lambda| > |\mu| \geq 0,
\]
the $p_i(n)$ are polynomials, $l$ is a positive integer, and $n-i = 0$ when $i > n$.

Pairing this fact with the statement and proofs of Lemma \ref{lemJordanBlockLimitOne} and Lemma \ref{lemJordanBlockLimitTwo}, it is clear that $|J^nx|/|(J^nx)_1|\rightarrow (1,0,\dots0)^\top$ as $n \rightarrow \infty$, where $(J^n_x)_1$ labels the first component of $J^nx$. And thus, orbits under iteration of $J$ induce a sequence in real projective space converging to $\PP P$ where $P$ is the one or two dimensional subspace invariant under $J_1$, following the arguments of the proofs in Lemmas \ref{lemJordanBlockLimitOne} and \ref{lemJordanBlockLimitTwo}.
\end{proof}

With these lemmas in hand, we move to prove the theorems stated in the introduction of this paper.

\section{Proofs of Theorems}\label{sectionTheorems}

We begin with the proof of Theorem \ref{thmOne}.

\begin{proof}[Proof of Theorem \ref{thmOne}]
Consider the following algorithm deciding instances $(A, x, U)$ of the Orbit Problem satisfying the assumptions of the Theorem. If the given instance is degenerate, reduce it to a finite set of non-degenerate instances and solve each independently using the following algorithm.

\begin{enumerate}
\item Using Lemma \ref{lemmaReduce}, transform $A$ into its Jordan canonical form $J = J(A)$, and use the invertible matrices $Q$ used in obtaining the Jordan canonical form to transfer $x$ and $U$ to the same basis. We now suppose $x$ and $U$ are presented in the new basis without updating their label.

During the computation of $J$, collect the eigenvalues $\lambda_1,\dots,\lambda_m$ of $A$, organized so that $|\lambda_1|\geq \cdots \geq |\lambda_m|$. Suppose $\lambda_1,\dots,\lambda_r$, $r \leq m$ are the eigenvalues of maximal modulus. Every eigenvalue $\lambda_i$ will correspond to one or more Jordan blocks composing $J$, where the number of Jordan blocks associated with an eigenvalue is determined by the algebraic and geometric multiplicity of the eigenvalue, along with the elementary divisors of the characteristic polynomial of $A$.

Let $P_A$ label the minimal polynomial of $A$. Then note that the roots of $P_A$ are the eigenvalues of $A$, and the multiplicity of each root in $P_A$ is the dimension of the \emph{largest} Jordan block associated with the respective root in the Jordan canonical form of $A$. This is a trivial consequence of the definition of the minimal polynomial of $A$.
However here we use the \emph{real} Jordan canonical form, and thus the complex roots of the minimal polynomial will correspond to Jordan blocks of dimension twice their multiplicity.

Following Lemma \ref{lemJordanBlockLimitThree}, the Jordan blocks associated with the $r$ eigenvalues of largest modulus dictate the asymptotic behavior of the orbit $\{\PP J^n x\}_{n=0}^\infty$. More specifically, by Lemma \ref{lemJordanBlockLimitTwo}, of those $r$ dominating eigenvalues, those corresponding to the Jordan blocks of largest dimension dominate the asymptotic behavior (up to a possible rescaling of $1/2$ for the complex eigenvalues due to the real Jordan canonical form). Precisely, by Lemma \ref{lemJordanBlockLimitTwo} the orbit $\{\PP J^nx\}_{n=0}^\infty$ converges to a subset $\PP P$ determined by the largest eigenvalues, and of those the Jordan blocks of highest dimension corresponding to such eigenvalues.

By assumption of Theorem \ref{thmOne}, there are $p\leq r$ distinct eigenvalues of maximal modulus with highest multiplicity in $P_A$.
As such, by Lemma \ref{lemJordanBlockLimitTwo} and the non-degeneracy assumption, there is a single subspace $P$ of dimension $p$ such that $d(\PP J^n x, \PP P) \rightarrow 0$ as $n \rightarrow \infty$ whenever $x$ has non-zero components with respect to the Jordan blocks composing $J$ --- which is ensured by the non-triviality assumption placed on $x$ in the theorem statement.

Return the subspace $P=P(J, x)$.

\item The subspace $P$ is taken to be presented as a set of rational basis vectors: it is computable given $x$, along with determining the position of the Jordan blocks of largest dimension corresponding to the eigenvalues of maximal modulus in the matrix $J$. This fact follows from the proofs of Lemmas \ref{lemJordanBlockLimitOne} and \ref{lemJordanBlockLimitTwo}; that is, $P$ is simply a direct sum of either standard basis vectors in the Jordan basis, or invariant subspaces described by rational linear combinations of standard basis vectors, the coefficients of which are effectively computable from the entries of $x$ as commented in Lemma \ref{lemJordanBlockLimitTwo}.

Compute $P \cap U$. If $P \cap U = \{0\}$, continue. Otherwise halt and terminate execution of the algorithm; this instance of the Orbit Problem cannot be decided by this algorithm.

\item We have $P \cap U =\{0\}$. Then using the minimal angle algorithm following from the singular value decomposition as discussed in Section \ref{secAnglesBetweenFlats}, compute the quantity
\[
\sigma_{max} = \Gamma(P, U).
\]
where the value $\Gamma(P, U)$ corresponds to the cosine of the minimal angle between $P$ and $U$.

Since $P \cap U = \{0\}$, it follows that $\sigma_{max} < 1$.
Let $\zeta = 1 - \sigma_{max}$. Then there must exist an $\epsilon = \epsilon(\zeta)>0$ such that $d(\PP P, \PP U) = \epsilon$.

\item Begin iterating $J$ over $x$. Every iteration, compute $\Gamma(J^nx, P)$, and check if $J^nx \in U$. If there is an $n$ such that $J^nx \in U$, halt: the orbit intersects the target set in this instance of the Orbit Problem. 

Otherwise, continue until $\Gamma(J^nx, P) > \sigma_{max}$: since $\Gamma(J^nx, P)\rightarrow 1$ as $n \rightarrow \infty$, by Lemma \ref{lemAngleConverge} $d(\PP J^nx, \PP P)\rightarrow 0$ as $n \rightarrow \infty$. But there is an $\epsilon > 0$ such that $d(\PP P, \PP U) = \epsilon$, and hence there is an $N \in \NN$ such that for all $n \geq N$, $\PP J^nx \not\in \PP U$. Hence, there is an $N'=N'(\sigma_{max}) \in \NN$ such that $J^n x \not\in U$ for all $n > N'$.
\end{enumerate}
\end{proof}

We now prove Theorem \ref{thmDimOne}.

\begin{proof}[Proof of Theorem \ref{thmDimOne}]
As usual, if the instance is degenerate we can reduce to a finite number of non-degenerate instances to check. Hence we assume non-degeneracy.
We begin by arguing the case when $A$ is nonsingular.

Take $(A, x, U)$ is a non-degenerate instance of the higher-dimensional Orbit Problem with $A$ nonsingular and $U$ of dimension one.
In addition, recall that it is assumed that the modulus of the eigenvalues of $A$ are all distinct, up to complex conjugate pairs.
Note that immediate from Theorem \ref{thmOne} and its proof, if the orbit $\{A^nx\}_{n=0}^\infty$ approaches a subspace $W$ after projecting onto real projective space, and $W \cap U = \{0\}$, then the instance is decidable.
Moreover, the condition that $\text{dim}(U) = 1$ enforces either $U \subseteq W$ or $U \cap W = \{0\}$.
Hence, we are left to consider the case when $U \subseteq W$.

$A$ is taken to be nonsingular with eigenvalues of distinct modulus up to conjugate pairs, so the dimension of the attracting subspace $W$ is strictly less than that of the ambient space, unless $d=2$ with $A$ admitting a complex conjugate eigenvalue pair. But in such a situation $A$ merely acts as a rational or irrational rotation of the plane up to scaling, and with $U$ one-dimensional this problem is well known to be decidable.

Suppose $x \not\in W$: then since $A$ is nonsingular $A^n x \not\in U$ for all $n$ since $W$ is invariant under $A$. If $x \in W$, then consider the reduced, lower-dimensional system $A:W \rightarrow W$ and induct until one of the above conditions is met or $\text{dim}(W) = 1$ (halt and return ``yes") or $\text{dim}(W) = 2$ (consider the circle rotation). Notice we can continue this induction since the eigenvalues of $A$ have distinct modulus.

Finally, note that if $A$ is singular, decidability is proved via essentially identical argument. Namely, in the singular case $A$ collapses points to some subspace $Q$, which is readily computed. Check if $x$ or $Ax \in U$. If not, via basic argument one may consider the restricted system $A|_Q$ with $U \cap Q$, and ask the orbit problem in this lower-dimensional setting, then proceed as before.
\end{proof}

The dynamical and geometric nature of the techniques used here provides simple proofs of other old results, such as the following concerning the decidability of Skolem's Problem in instances of a dominating real root. We use the proof of this result to then aid in proving Theorem \ref{thmSkolemTwo}.

\begin{proposition}\label{thmSkolemOne}
Let $\{u_n\}_{n=0}^\infty$ be an order $d$ non-degenerate linear recurrence sequence with distinct characteristic roots $\lambda_1,\dots,\lambda_m$, $m \leq d$, ordered so that $|\lambda_1|\geq \cdots \geq |\lambda_m|$. Suppose the initial terms form a vector $x \in \QQ^d$ non-trivial with respect to the companion matrix $A$ of the sequence. Then, if $|\lambda_1|> |\lambda_2|$, it is decidable whether the sequence has a zero term.
\end{proposition}
\begin{proof}
Let $A \in \QQ^{d \times d}$ label the companion matrix of the LRS $\{u_n\}$. Let $x$ denote the non-trivial $d$ dimensional vector of initial terms of the LRS. Then, as noted in Section \ref{secLinearRecurrence}, iteration of $A$ over $x$ ``shifts" the elements of the LRS through $x$, so that when $x = (u_d,u_{d-1},\dots,u_2,u_1)^\top$, $Ax = (u_{d+1}, u_d,\dots,u_3,u_2)^\top$.

Let $E_1,E_2,\dots,E_d$ denote the $d$ coordinate subspaces of $\RR^d$ where the elements of $E_i$, $i=1,\dots,d$, have a $0$ in their $i$th component. Then $AE_d \subseteq E_{d-1},AE_{d-1}\subseteq E_{d-2},\dots,AE_{2}\subseteq E_1$. Moreover, $E_1\cap E_2 \cap\cdots \cap E_d =\{0\}$. Hence, for any $y \neq 0 \in \QQ^d$, $y \not\in E_1 \cap \cdots \cap E_d$, implying there is an $E_i$ such that $y \not \in E_i$.

The LRS $\{u_n\}$ is taken to have a single dominating real root $\lambda_1$. Hence $A$ has a single dominating real root. Taking $A$ to its Jordan canonical form and back via a change of basis, we see by consequence of Lemmas \ref{lemJordanBlockLimitOne}, \ref{lemJordanBlockLimitTwo}, \ref{lemJordanBlockLimitThree} that there is a line spanned by some $y \neq 0 \in \QQ^d$, such that $d(\PP A^nx, \PP y) \rightarrow 0$ as $n\rightarrow \infty$. But there is an $E_i$ such that $y \not \in E_i$, i.e. $\text{span}(y) \cap E_i = \{0\}$. Then, appealing to the statement and proof of Theorem \ref{thmOne}, it can be decided in a finite number of iterations of $A$ over $x$ whether the LRS has a zero.
\end{proof}

We now prove Theorem \ref{thmSkolemTwo}.

\begin{proof}[Proof of Theorem \ref{thmSkolemTwo}]
The proof is a trivial extension of the proof of Proposition \ref{thmSkolemOne} with Lemma \ref{lemJordanBlockLimitTwo}. Let $A \in \QQ^{d \times d}$ label the companion matrix of the LRS, and let $P_A$ denote the minimal polynomial of $A$.

In Proposition \ref{thmSkolemOne}, there is a single dominating characteristic root, while in the case of Theorem \ref{thmSkolemTwo} there are $r \geq 1$ characteristic roots of maximal modulus. However, by assumption, there is a single real root, $\lambda_1$, whose multiplicity in the minimal polynomial $P_A$ of $A$ is larger than the multiplicity of the other dominating roots. This implies that, of the $r$ dominating roots, the Jordan block of largest dimension which is associated with $\lambda_1$ dictates the asymptotic behavior. Then by Lemma \ref{lemJordanBlockLimitTwo}, after projecting to real projective space, orbits approach the projection of a line, and hence must eventually be bounded away from at least one of the $E_i$ after a finite number of iterations, where the $E_i$ are as defined in the proof of Proposition \ref{thmSkolemOne}.
\end{proof}

We now move to consider the continuous Orbit Problem, studying flows rather than maps. It is in this setting that the foregoing development of dynamics in projective space is most useful: as in the discrete case, every flow is asymptotic to some stationary point or cycle. Then, we have a way of bounding the location of flows in projective space, which aids in decidability; particularly because when considering flows, traversal from one side of a boundary to another implies intersecting the boundary.

\begin{proof}[Proof of Theorem \ref{thmContTwo}.]
In the case $A$ has a complex-conjugate pair of eigenvalues, the projection of the flow onto projective space is simply a cycle. Hence the cycle necessarily intersects with the target line, which is a point in $\mathbb{RP}$, infinitely often. This holds independently of whether the real component of the complex eigenvalue pair is negative, positive, or zero.

Whenever $A$ has two distinct real eigenvalues or a single repeated real eigenvalue, the flow is stationary when the eigenvalue(s) are zero, and otherwise if not zero the flow approaches a stationary point $p \in \mathbb{RP}$, computable from $A$ by translating to the Jordan normal form. Since the target space $U$ has rational entries, it is decidable whether $p = \mathbb{P}U$. If $p = \mathbb{P}U$, and $x(0) \neq U$, then $x(t) \not\in U$ for all $t\in \mathbb{R}^+$, since $U$ is a stationary point in projective space and flows are unique. Otherwise, when $p \neq \mathbb{P}U$, check if $\mathbb{P}x(0)$ and $p$ are on the same side of $\mathbb{P}U$. If so, $x(t) \not\in U$ for all $t \in \mathbb{R}^+$. If not, there exists a $t \in \mathbb{R}^+$ such that $x(t) \in U$ because $\mathbb{P}x(t)$ is forced to cross $\mathbb{P} U$ on approach to $p$.
\end{proof}

\section{Concluding remarks}\label{secConclusion}

This paper does reach the limits of the techniques presented here, nor do we apply this machinery to every problem vulnerable to our methods. Indeed, the results of this paper indicate that deepening our understanding of dynamical systems in real projective space $\RR\PP^{d-1}$ with maps induced by matrices in $GL(d, \RR)$ provides insight into the higher-dimensional Orbit Problem, and more generally termination problems, by way of the arguments given in this paper. 
And, possibly, the geometric and dynamical mechanisms penetrating the Orbit Problem may link to the algebraic and number-theoretic structures traditionally employed in nontrivial ways. In particular, we believe that mixing the general geometric and dynamical structures provided here with the finer algebraic and number-theoretic tools traditionally used can lead to additional breakthroughs in this area.

To this end, we identify a number of different directions that may be profitable to explore, given the methods presented here. We begin by considering the following: every linear system in $\RR^d$ induces a dynamical system on the set of $k$-dimensional subspaces of $\RR^d$ --- the Grassmannians. A possibly rewarding next step could be to generalize the results of this paper further by studying such induced dynamics on Grassmannians, where a more detailed understanding of the dynamics of this kind can result in stronger results toward the Orbit Problem.

Next, we reflect that the essential reason why we require that $W \cap U = \{0\}$ in order to have decidability in Theorem \ref{thmOne}, follows from the main observation this paper makes, which is that orbits approach the projection of $W$ in real projective space, i.e. the angles between orbits and $W$ monotonically approaches zero, and as a consequence whenever $W \cap U = \{0\}$ orbits can be bounded away from $U$ in finite time. Fortunately, if $A, x$, and $U$ are randomly generated by drawing entries at random from some finite set, when $A$ has $r$ dominating roots, $p$ of which have highest multiplicity in the minimal polynomial of $A$, and $U$ dimension $\leq d-p$, then with overwhelmingly high probability $W \cap U = \{0\}$. Thus, arguing from this intuitive level, Theorem \ref{thmOne} decides a ``large" class of instances.

Nonetheless, difficulty arises when the intersection of $W$ and $U$ is nontrivial. In such cases, the orbits approach $U$, or periodically get arbitrarily close to $U$. But this is when the basic argument employed in this paper can no longer be exploited. To overcome this barrier, finer methods must be used. In particular, if it is better understood \emph{how} induced orbits in projective space approach attracting sets, then this will immediately translate to deciding the Orbit Problem. Fortunately, the attracting sets in projective space have a nice structure: they are either fixed points under the induced map or contained in closed loops. To this end, possibly more can be said in the continuous case, where we consider flows.

Although, we remark that when we have containment $W \subseteq U$, we do obtain something close to decidability, since orbits $\{\PP A^nx\}_{n=0}^\infty$ approach $\PP U$ as $n\rightarrow \infty$, leading to decidability when we work with notions such as ``pseudo-orbits," already explored in literature \cite{akshay2022robustness, CostaPseudoReachabilityDiag,dcostaPseudoSkolem}. 
As such, when $W \subseteq U$ we obtain ``pseudo-decidability:" orbits converge toward the target set. And, following the proof of Theorem \ref{thmDimOne}, if $A$ is nonsingular then we obtain decidability when $x \not \in W$.
Indeed, in the context of program verification and the Termination Problem, this asymptotic behavior indicates such orbits could enter the error state with sufficiently large perturbation. To this end, perhaps progress can be made on deciding cases in which $W \subseteq U$.

\bibliographystyle{abbrv}
\bibliography{references}

\end{document}